\documentclass[11pt,prb,onecolumn,showpacs,amssymb,superscriptaddress]{revtex4-2}
\usepackage{graphicx}
\usepackage{amsmath}
\usepackage{amsfonts}
\usepackage{amsthm}
\usepackage{amssymb}
\usepackage{amsbsy}
\usepackage{wasysym}
\usepackage{bm}
\usepackage{mathrsfs}
\usepackage{color}
\usepackage{hyperref}
\usepackage{fullpage}
\usepackage{bbm}
\usepackage{physics}
\hypersetup{
    colorlinks=true,
    linkcolor=red,        
    citecolor=blue,
    breaklinks=true,
    urlcolor=blue
}
\newcommand{\poly}{\text{poly}}

\newtheorem{theorem}{Theorem}
\newtheorem{prop}{Proposition}
\newtheorem{corollary}{Corollary}
\newtheorem{lem}{Lemma}

\renewcommand{\S}{\hat{\mathbf{S}}}

\begin{document}

\title{The classical limit of Quantum Max-Cut}

\author{Vir B. Bulchandani}
\affiliation{Department of
Physics and Astronomy, Rice University, 6100 Main Street
Houston, TX 77005, USA}

\author{Stephen Piddock}
\affiliation{Department of Computer Science, Royal Holloway, University of London, UK}
\begin{abstract}
It is well-known in physics that the limit of large quantum spin $S$ should be understood as a semiclassical limit. This raises the question of whether such emergent classicality facilitates the approximation of computationally hard quantum optimization problems, such as the local Hamiltonian problem. We demonstrate this explicitly for spin-$S$ generalizations of Quantum Max-Cut ($\textsc{QMaxCut}_S$), equivalent to the problem of finding the ground state energy of an arbitrary spin-$S$ quantum Heisenberg antiferromagnet ($\textsc{QHA}_S$). We prove that approximating the value of $\textsc{QHA}_S$ to inverse polynomial accuracy is QMA-complete for all $S$, extending previous results for $S=1/2$. We also present two distinct families of classical approximation algorithms for $\textsc{QMaxCut}_S$ based on rounding the output of a semidefinite program to a product of Bloch coherent states. The approximation ratios for both our proposed algorithms strictly increase with $S$ and converge to the Bri{\"e}t-Oliveira-Vallentin approximation ratio $\alpha_{\textsc{BOV}} \approx 0.956$ from below as $S \to \infty$. 
\end{abstract}

\maketitle
{
\hypersetup{linkcolor=black}
\tableofcontents
}
\section{Introduction}
A remarkable prediction of classical computer science is the existence of fundamental limits on our ability to solve difficult optimization problems. For concreteness, consider an NP-hard constraint satisfaction problem such as Max-Cut~\cite{karp2010reducibility}. The effectiveness of a possibly random approximation algorithm for this problem, which yields an estimate $W$ for the optimal value $Z^*$ for a given problem instance, can be quantified by an ``approximation ratio'' $\alpha$ such that
\begin{equation}
\alpha Z^* \leq \mathbb{E}[W] \leq Z^*
\end{equation}
for all problem instances. Assuming that P $\neq$ NP, the classical PCP theorem~\cite{arora1998proof} implies upper bounds $\alpha \leq \alpha^* <1$ on the best possible approximation ratio that can be achieved in polynomial time, predicting~\cite{trevisan2000gadgets,haastad2001some} $\alpha^* = 16/17$ for Max-Cut. The stronger assumptions of the unproven Unique Games Conjecture~\cite{khot2007optimal} (UGC) imply a tighter upper bound $\alpha^* = 0.878...$, which coincides with the best polynomial-time approximation ratio known for this problem. The latter is due to Goemans and Williamson~\cite{goemans1995improved}, whose algorithm consists of randomly rounding the output of a semidefinite program (SDP) to a feasible solution. Analogous SDP-based algorithms and proofs of their optimality (assuming the UGC) can be formulated for arbitrary classical constraint satisfaction problems~\cite{raghavendra2008optimal}. Thus the computational hardness of approximating such problems is rather well understood, pending the resolution of the UGC.

Convincing quantum analogues of these results are yet to be discovered and the existence of a comparably strong quantum PCP theorem remains unclear~\cite{aharonov2013quantum}. A physically natural class of quantum optimization problems is furnished by the $k$-local Hamiltonian problem, for which various works have proposed more-or-less general polynomial-time algorithms~\cite{bansal2008classical,Gharibian_2012,brandao,bravyi2019approximation}. More recently, the scope of such efforts has narrowed to focus on the so-called ``Quantum Max-Cut'' problem~\cite{GP19,hwang2022unique,parekh2022optimal,king2022improved,lee2022optimizing,takahashi20232,watts2023relaxations,lee2024improved} ($\textsc{QMaxCut}_{1/2}$) proposed by Gharibian and Parekh~\cite{GP19} (GP), which is an appealingly simple quantum generalization of classical Max-Cut and equivalent to finding the ground state of a spin-$1/2$ Heisenberg antiferromagnet with arbitrary two-body couplings ($\textsc{QHA}_{1/2})$. Approximating the latter to inverse polynomial accuracy is known to be QMA-complete~\cite{piddock2015complexity} and thus $\textsc{QMaxCut}_{1/2}$ defines a QMA-hard maximization problem. Assuming the UGC, there is a conjectured upper bound~\cite{hwang2022unique} on the approximation ratio achievable for $\textsc{QMaxCut}_{1/2}$ on a classical computer in polynomial time, given by the classical Bri{\"e}t-Oliveira-Vallentin~\cite{briet2010positive} (BOV) approximation ratio $\alpha_{\textsc{BOV}} \approx 0.956$. At the same time, various instances of $\textsc{QHA}_{1/2}$ have been studied in condensed matter physics for nearly a century, ranging from highly structured examples solvable by Bethe ansatz techniques~\cite{bethe1931theorie,gaudin1976diagonalisation,haldane1988exact,shastry} to disordered examples with random couplings~\cite{Dasgupta,bhatt1982scaling} that capture various aspects of spin-glass physics. 

We note that the goals of physicists and of computer scientists working on such problems have historically been rather different. The physics literature seeks to construct trial wavefunctions that approximate ground states well in practice, without necessarily attempting to construct these wavefunctions in polynomial time or to establish rigorous bounds on their approximation ratios. Nevertheless, some of these constructions are remarkably accurate~\cite{Huse}. The computer science literature has largely focused on constructing product-state approximations, or small variations thereof, in polynomial time, and tends to obtain poor but rigorous approximation ratios. Such results nevertheless limit the scope of a possible quantum PCP theorem~\cite{brandao,Gharibian_2012}.

We believe that finding ``optimal'' algorithms for quantum local Hamiltonian problems, that are optimal in the sense that the Goemans-Williamson algorithm is conjectured~\cite{khot2007optimal} to be optimal for classical Max-Cut, will require bridging this gap between physically sensible ground states and rigorous proofs. As a step in this direction, this paper introduces a family of spin-$S$ generalizations of Quantum Max-Cut ($\textsc{QMaxCut}_S$). These are qudit generalizations of $\textsc{QMaxCut}_{1/2}$ with onsite Hilbert space dimension $d=2S+1$. (Note that these are distinct from qudit generalizations~\cite{piddock2021universal,carlson2023approximation} of Quantum Max-Cut that involve the fundamental representation of $SU(d)$ for $d \geq 2$.) Our motivation for introducing these optimization problems is twofold: first, the limit of large spin $S$ is well-known to define a semiclassical limit, in which product states yield arbitrarily good approximations to the true ground state as $S \to \infty$~\cite{lieb1973classical}. Thus existing product-state algorithms become better physically motivated as the spin $S$ increases, and should correspondingly exhibit better approximation ratios. We show that this is indeed the case. Second, generalizing semidefinite-programming-based algorithms for Quantum Max-Cut to spin $S>1/2$ requires properly accounting for the $SU(2)$ symmetry of the Heisenberg model, which has proved to be an increasingly useful heuristic for achieving better approximation algorithms when $S=1/2$~\cite{king2022improved,takahashi20232,watts2023relaxations}.

The paper is structured as follows. In the remainder of this introduction, we introduce Spin-$S$ Quantum Max-Cut and the tools that we will need for its analysis. We then derive a simple upper bound on the product-state approximation ratio, before presenting two distinct classical algorithms based on rounding SDPs to products of Bloch coherent states. Our first algorithm uses the classical Max-Cut SDP, and we estimate its approximation ratio using some inequalities due to Lieb~\cite{lieb1973classical}. Our second algorithm generalizes the Gharibian-Parekh SDP~\cite{GP19,hwang2022unique} to spin $S>1/2$. We find that the latter algorithm outperforms the na{\"i}ve semiclassical estimate based on Lieb's inequalities, and is conjecturally optimal in the same sense as the $S=1/2$ GP algorithm, i.e. achieves its integrality gap~\cite{hwang2022unique}. We show that the approximation ratios for both our proposed algorithms converge to $\alpha_{\textsc{BOV}}$ from below as $S \to \infty$. Finally, we prove that computing the value of Spin-$S$ Quantum Max-Cut to inverse polynomial accuracy is QMA-complete.

\subsection{Spin-$S$ Quantum Max-Cut}
We define the spin-$S$ Quantum Max-Cut Hamiltonian on an undirected, weighted graph $G = (V,E,w)$ with edge weights $w_{ij} \geq 0$ to equal
\begin{equation}
\hat{H}_{\textsc{QMC}_S}(G) = \frac{1}{2} \sum_{\{i,j\} \in E} w_{ij}\left(\hat{\mathbbm{1}} -\frac{1}{S^2} \hat{\mathbf{S}}_i \cdot \hat{\mathbf{S}}_j\right).
\end{equation}
Here the on-site Hilbert space dimension $d = 2S+1$ and the total number of vertices or sites $N = |V|$. The operators $\hat{\mathbf{S}}_j = (\hat{S}^{1}_j,\hat{S}^2_j,\hat{S}^3_j)$ are the standard~\cite{weinberg2015lectures} spin-$S$ operators acting at each site and satisfy the usual spin commutation relations
\begin{equation}
[\hat{S}_{j}^\alpha,\hat{S}_k^\beta] = i\epsilon^{\alpha \beta \gamma} \delta_{jk} \hat{S}_k^\gamma,
\end{equation}
while the dot product
\begin{equation}
\label{eq:dot}
\hat{\mathbf{S}}_i \cdot \hat{\mathbf{S}}_j = \sum_{\alpha=1}^3 \hat{S}_i^\alpha \hat{S}_j^\alpha
\end{equation}
and satisfies
\begin{equation}
\label{eq:Ssquared}
\| \hat{\mathbf{S}}_j\|^2 = \hat{\mathbf{S}}_j \cdot \hat{\mathbf{S}}_j = S(S+1) \hat{\mathbbm{1}}
\end{equation}
at each site. In what follows, norms and dot products will always indicate either norms and dot products of real 3-vectors with respect to the Euclidean metric, or norms and dot products of 3-vectors of operators that inherit this metric, as in Eqs. \eqref{eq:dot} and \eqref{eq:Ssquared}. The optimization problem Spin-$S$ Quantum Max-Cut ($\textsc{QMaxCut}_S$) of interest in this work consists of finding the largest eigenvalue
\begin{equation}
\textsc{QMaxCut}_S(G) = \max_{\substack{|\psi\rangle \in \mathbb{C}^{d^N} \\ \langle \psi|\psi\rangle = 1}} \, \langle \psi | \hat{H}_{\textsc{QMC}_S}(G) | \psi \rangle.
\end{equation} 
We will refer to the latter as the ``value'' of $\textsc{QMaxCut}_S$. When $S=1/2$, this recovers the optimization problem usually called Quantum Max-Cut~\cite{Gharibian_2012}. We note that by Eq. \eqref{eq:Ssquared} and the rules~\cite{weinberg2015lectures} for adding quantum angular momenta,
\begin{equation}
\label{eq:defform}
\langle \psi | \hat{H}_{\textsc{QMC}_S}(G) | \psi \rangle = \frac{1}{2} \sum_{\{i,j\} \in E} w_{ij}\left(2+\frac{1}{S} -\frac{1}{2S^2} \langle \psi | \|\hat{\mathbf{S}}_i + \hat{\mathbf{S}}_j\|^2|\psi\rangle\right) \geq 0
\end{equation}
for all states $|\psi\rangle$. It will be useful to define the complementary problem $\textsc{QHA}_S$ of minimizing the energy of the corresponding spin-$S$ Heisenberg antiferromagnet
\begin{equation}
\hat{H}_{\textsc{QHA}_S}(G) = \frac{1}{2S^2}\sum_{\{i,j\}\in E} w_{ij} \hat{\mathbf{S}}_i \cdot \hat{\mathbf{S}}_j,
\end{equation}
with value
\begin{equation}
\textsc{QHA}_S(G) = \min_{|\psi\rangle \in \mathbb{C}^{d^N} : \langle \psi|\psi\rangle = 1} \, \langle \psi | \hat{H}_{\textsc{QHA}_S}(G) | \psi \rangle.
\end{equation}
The values of these two optimization problems are related in the obvious manner, namely
\begin{equation}
\label{eq:qrelation}
\textsc{QMaxCut}_S(G) = \frac{1}{2}\sum_{\{i ,j\} \in E} w_{ij} - \textsc{QHA}_S(G).
\end{equation}

\subsection{The product-state value and Bloch coherent states}
The ``product-state value'' of $\textsc{QMaxCut}_S$ is given by its maximum over product states, explicitly
\begin{equation}
\textsc{Prod}_S(G) = \max_{\substack{|\psi_i\rangle \in \mathbb{C}^{d} \\ \langle \psi_i | \psi_i \rangle = 1}} \frac{1}{2}\sum_{\{i,j\} \in E} w_{ij} \left(1 - \frac{1}{S^2}\langle \psi_i | \hat{\mathbf{S}}_i |\psi_i \rangle \cdot \langle \psi_j | \hat{\mathbf{S}}_j |\psi_j \rangle\right).
\end{equation}
In fact, this optimization problem is lower in dimension than it appears. To see this, it is helpful to introduce real 3-vectors $\mathbf{u}_i \in \mathbb{R}^3$ with components
\begin{equation}
u_i^\alpha = \frac{1}{S} \langle \psi_i | \hat{S}_i^\alpha | \psi_i\rangle, \quad \alpha=1,2,3.
\end{equation}
By an appropriate statement of the uncertainty principle for quantum spins, such vectors satisfy~\cite{delbourgo1977maximum}
\begin{equation}
\label{eq:quditstate}
\|\mathbf{u}_i\| \leq 1,
\end{equation}
with equality iff $|\psi_i\rangle$ is a highest-weight state of the operator $\hat{\mathbf{S}} \cdot \mathbf{\Omega}$ along some axis $\mathbf{\Omega} \in S^2$. It follows that
\begin{equation}
\label{eq:generalprod}
\textsc{Prod}_S(G) = \max_{\mathbf{u}_i \in B^{3}} \frac{1}{2}\sum_{\{i,j\} \in E} w_{ij}\left(1 - \mathbf{u}_i \cdot \mathbf{u}_j\right)
\end{equation}
without loss of generality, where the unit 3-ball $B^3 = \{\mathbf{u} \in \mathbb{R}^3 : \| \mathbf{u}\| \leq 1\}$. Thus we have reduced an optimization over $2d-2$ real parameters per site to an optimization over three real parameters per site.

In fact, one can reduce the dimensionality of this optimization problem further to two real parameters per site by introducing so-called Bloch coherent states~\cite{lieb1973classical} that saturate the inequality Eq. \eqref{eq:quditstate}. We write Bloch coherent states as $|\mathbf{\Omega}\rangle$, where $\mathbf{\Omega}\in S^2 = \partial B^3$ lies on a unit 2-sphere. The defining property of these states for our purposes is the formula
\begin{equation}
\langle \mathbf{\Omega} | \hat{\mathbf{S}} | \mathbf{\Omega} \rangle = S \mathbf{\Omega}.
\end{equation}
Then a simple but important observation is that the product state value Eq. \eqref{eq:generalprod} is attained by a product of Bloch coherent states. This is the content of the following Proposition.
\begin{prop}
\label{prop:bloch}
The product state value
\begin{equation}
\textsc{Prod}_S(G) = \max_{\mathbf{\Omega}_i \in S^2} \frac{1}{2}\sum_{\{i,j\} \in E} w_{ij}\left(1 - \mathbf{\Omega}_i \cdot \mathbf{\Omega}_j\right).
\end{equation}
\end{prop}
\begin{proof}
It suffices to show that one can always replace $\mathbf{u}_i \in B^3$ with a unit vector $\mathbf{\Omega}_i \in S^2$ such that the sum
\begin{equation}
\mathcal{E} = \sum_{\{i,j\}\in E} w_{ij} \mathbf{u}_i \cdot \mathbf{u}_j
\end{equation}
does not increase under this replacement. To see this, note that for each vertex $i \in V$, we can write $\mathcal{E} = \mathcal{E}_i + \bar{\mathcal{E}}_i$, where
\begin{equation}
\mathcal{E}_i = \sum_{\{j \in V : \{i,j\}\in E\}} w_{ij} \mathbf{u}_i \cdot \mathbf{u}_j = \mathbf{u}_i \cdot \mathbf{w}_i
\end{equation}
and $\mathbf{w}_i = \sum_{\{j \in V : \{i,j\}\in E\}} w_{ij} \mathbf{u}_j$. Then 
\begin{equation}
\mathcal{E}  = \mathbf{u}_i \cdot \mathbf{w}_i + \bar{\mathcal{E}}_i \geq -\|\mathbf{w}_i\| + \bar{\mathcal{E}}_i
\end{equation}
and this inequality is saturated by replacing $\mathbf{u}_i$ with $\mathbf{\Omega}_i$, where we define $\mathbf{\Omega}_i = - \frac{\mathbf{w}_i}{\|\mathbf{w}_i\|}$ at each vertex if $\mathbf{w}_i \neq 0$ and $\mathbf{\Omega}_i$ to be an arbitrary unit vector otherwise. The result follows upon starting from an optimal choice of $\mathbf{u}_i$ and defining $\mathbf{\Omega}_i$ in this manner at each vertex in turn.
\end{proof}
Since by Proposition \ref{prop:bloch} the product state value is independent of the spin $S$, we will henceforth write it as $\textsc{Prod}(G)$. The latter optimization problem is also known as Rank-3 Max-Cut~\cite{hwang2022unique}, and the Bri{\"e}t-Oliveira-Vallentin algorithm~\cite{briet2010positive} yields an approximation ratio $
\alpha_{\textsc{BOV}} \approx 0.956$ for this problem, where $\alpha_{\textsc{BOV}}$ is defined precisely in Eq. \eqref{eq:alphaBOVdef}.

It will occasionally be useful to define the corresponding minimization problem $\textsc{CHA}$ for a classical Heisenberg antiferromagnet
\begin{equation}
H_{\textsc{CHA}}(G,\vec{\mathbf{\Omega}}) = \frac{1}{2}\sum_{\{i,j\} \in E} w_{ij} \mathbf{\Omega}_i \cdot \mathbf{\Omega}_j,
\end{equation}
where we write $\vec{\mathbf{\Omega}}= (\mathbf{\Omega}_1,\mathbf{\Omega}_2,\ldots,\mathbf{\Omega}_N)$ as shorthand for an $N$-tuple of $S^2$-valued spins $\mathbf{\Omega}_i \in S^2$. The value of this minimization problem is given by
\begin{equation}
\textsc{CHA}(G) =  \min_{\mathbf{\Omega}_i \in S^2}  \, H_{\textsc{CHA}}(G,\vec{\mathbf{\Omega}}),
\end{equation}
and the analogue of Eq. \eqref{eq:qrelation} reads
\begin{equation}
\label{eq:crelation}
\textsc{Prod}(G) = \frac{1}{2}\sum_{\{i,j\} \in E} w_{ij} - \textsc{CHA}(G).
\end{equation}
The decision problem corresponding to determining $\textsc{Prod}(G)$ to inverse polynomial accuracy is contained in NP. The decision problem corresponding to exactly solving $\textsc{Prod}(G)$ or equivalently $\textsc{CHA}(G)$ was very recently shown to be NP-complete~\cite{kallaugher2024complexity} in the case where the weights $w_{ij}$ can be positive or negative.

\subsection{Bounding the product-state approximation ratio}

As a first illustration of the power of Bloch coherent states for understanding $\textsc{QMaxCut}_S$, we obtain an upper bound $\alpha^*(S)$ on the approximation ratio for $\textsc{QMaxCut}_S$ that can be achieved using product states, in the spirit of Gharibian and Parekh's analysis~\cite{GP19} for $S=1/2$, which found that $\alpha^*(S) = 1/2$.

By Proposition \ref{prop:bloch}, the optimal product-state approximation ratio for spin $S$ is given by
\begin{equation}
\alpha_{\textsc{Prod}}(S) = \inf_{G} \frac{\textsc{Prod}(G) }{\textsc{QMaxCut}_S(G)}.
\end{equation}
Following Gharibian and Parekh~\cite{GP19}, we obtain an elementary upper bound
\begin{equation}
\alpha_{\textsc{Prod}}(S) \leq \alpha^*(S) = \frac{\textsc{Prod}(G_0)}{\textsc{QMaxCut}_S(G_0)},
\end{equation}
from the graph $G_0=(\{1,2\}, \{\{1,2\}\},w_{12}=2)$ consisting of a single weight-2 edge connecting two vertices. We first note that by Eq. \eqref{eq:defform},
\begin{equation}
\langle \psi | \hat{H}_{\textsc{QMC}_S}(G_0) | \psi \rangle =  2+\frac{1}{S} -\frac{1}{2S^2}\langle \psi | \|\hat{\mathbf{S}}_{\mathrm{tot}}\|^2 |\psi\rangle  \leq 2+ \frac{1}{S}.
\end{equation} 
where $\hat{\mathbf{S}}_{\mathrm{tot}} = \hat{\mathbf{S}}_1+\hat{\mathbf{S}}_2$. Expanding $|\psi\rangle$ as usual~\cite{weinberg2015lectures} in the joint eigenbasis of $\hat{S}_{\mathrm{tot}}^3$ and $\|\hat{\mathbf{S}}_{\mathrm{tot}}\|^2$, this 
inequality is saturated by the state with total spin zero, $\|\hat{\mathbf{S}}_{\mathrm{tot}}\|^2|\psi\rangle =0$, and thus
\begin{equation}
{\textsc{QMaxCut}_S(G_0)} = 2+ \frac{1}{S}.
\end{equation}
Meanwhile, the product state value
\begin{equation}
 \textsc{Prod}(G_0) = \max_{\mathbf{\Omega}_i \in S^2 }\left(1 -\mathbf{\Omega}_1 \cdot \mathbf{\Omega}_2 \right) = 2
\end{equation}
by Proposition \ref{prop:bloch}, where the maximum is attained whenever $\mathbf{\Omega}_1 = -\mathbf{\Omega}_2$.
Thus the product state approximation ratio is bounded above by
\begin{equation}
\label{eq:upperbound}
\alpha^*(S) = \frac{2S}{2S+1}.
\end{equation}
When $S=1/2$, this recovers the Gharibian-Parekh bound, which was later found to be tight~\cite{parekh2022optimal}. Note that
\begin{equation}
\alpha^*(S) \to 1, \quad S \to \infty,
\end{equation}
which is indirect evidence that product-state approximations to $\textsc{QMaxCut}_S$ can perform arbitrarily well in principle as $S \to \infty$. In fact, the latter conclusion is immediate from Lieb's inequalities~\cite{lieb1973classical} (see Theorem \ref{thm:Lieb}).

\subsection{SDP relaxations}
\label{sec:SDP}
We now introduce two SDP relaxations of the above problems that will be used in this paper. The first is the Goemans-Williamson relaxation of classical Max-Cut~\cite{goemans1995improved} (which is also a relaxation of $\textsc{Prod}$~\cite{hwang2022unique})
\begin{equation}
\label{eq:SDPMC}
    \textsc{SDP}_{\textsc{MC}}(G) = \max_{\mathbf{y}_i \in S^{N-1}} \frac{1}{2}\sum_{\{i,j,\} \in E}\left(1- \mathbf{y}_i \cdot \mathbf{y}_j\right).
\end{equation}

The second SDP of interest is a Gharibian-Parekh-type~\cite{GP19} relaxation of $\textsc{QMaxCut}_S$, given by
\begin{equation}
\label{eq:SDPS}
\textsc{SDP}_S(G) = \max_{\mathbf{y}_i \in S^{N-1}} \frac{1}{2}\sum_{\{i,j,\} \in E}\left(1- \left(\frac{S+1}{S}\right)\mathbf{y}_i \cdot \mathbf{y}_j\right).
\end{equation}
We prove in Proposition \ref{prop:relax} that this is indeed a relaxation of $\textsc{QMaxCut}_S$. This ``spin-$S$ SDP'' recovers (a suitable statement of) the GP SDP~\cite{hwang2022unique,parekh2022optimal} when $S=1/2$, but differs from this SDP for $S > \frac{1}{2}$. Note that the objective function of the spin-$S$ SDP converges uniformly to the objective function of the Max-Cut SDP
as $S \to \infty$; this is how the SDP reflects the emergence of a semiclassical limit for large $S$.

To construct trial wavefunctions for $\textsc{QMaxCut}_S$ from these SDP relaxations, we will use the same rounding scheme in both cases. The first step is to randomly round each SDP output $\mathbf{y}_i$ to a unit $3$-vector $\mathbf{\Omega}_i$, via
\begin{equation}
\label{eq:round}
\mathbf{\Omega}_i = \frac{Z\mathbf{y}_i}{\|Z\mathbf{y}_i\|} \in S^2,
\end{equation}
where $Z$ is a $3$-by-$N$ random matrix with i.i.d. standard normal entries. The next step is to construct the product of Bloch coherent states
\begin{equation}
\label{eq:prodemb}
|\vec{\mathbf{\Omega}}\rangle = \otimes_{i\in V} |\mathbf{\Omega}_i\rangle.
\end{equation}

For the Max-Cut SDP Eq. \eqref{eq:SDPMC}, it is known~\cite{briet2010positive} that randomized rounding as in Eq. \eqref{eq:round} yields an $\alpha_{\textsc{BOV}}$-approximation to \textsc{Prod}, which has further been conjectured to be optimal~\cite{hwang2022unique}. We refer to this scheme for approximating \textsc{Prod} as the ``BOV algorithm''. In particular, we show in Corollary \ref{cor:BOV} that constructing product states Eq. \eqref{eq:prodemb} from the BOV algorithm and using Lieb's inequalities~\cite{lieb1973classical} to estimate the approximation ratio yields an 
\begin{equation}
\label{eq:alpha1}
\alpha_{\mathrm{L}}(S) = \left(\frac{S}{S+1}\right)^2 \alpha_{\textsc{BOV}}
\end{equation}
approximation to $\textsc{QMaxCut}_S(G)$.
\begin{figure}[t]
\centering\includegraphics[width=0.75\linewidth]{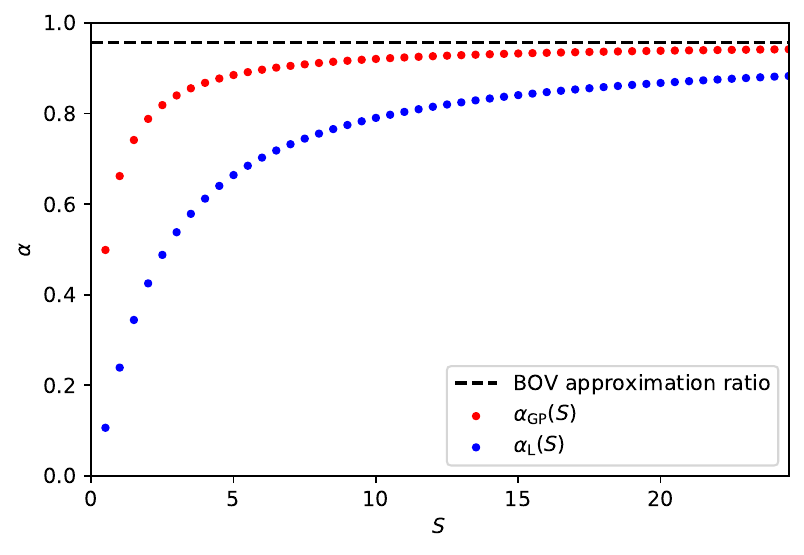}
    \caption{Comparison between the approximation ratios appearing in Eq. \eqref{eq:ineq} and the BOV approximation ratio. Convergence of the approximation ratios for both our proposed algorithms to the BOV value (black dashed line), as well as the superior performance of $\alpha_{\mathrm{GP}}(S)$ (red dots) compared to $\alpha_{\mathrm{L}}(S)$ (blue dots) are apparent.}
    \label{fig:approx}
\end{figure}
We can similarly round a solution of the spin-$S$ SDP to an approximation to $\textsc{QMaxCut}_S$ by applying Eqs. \eqref{eq:round} and \eqref{eq:prodemb} as above, to yield a distinct approximation ratio $\alpha_{\mathrm{GP}}(S)$, derived in Theorem \ref{thm:GP}. When $S=1/2$, this recovers the Gharibian-Parekh algorithm with $\alpha_{\mathrm{GP}}(1/2) \approx 0.498$. For general values of $S \geq 1/2$, we prove in Proposition \ref{prop:ineq} that
\begin{equation}
\label{eq:ineq}
\alpha_{\mathrm{L}}(S) < \alpha_{\mathrm{GP}}(S) < \alpha_{\textsc{BOV}},
\end{equation}
which together with Eq. \eqref{eq:alpha1}, implies that
\begin{equation}
\label{eq:ineqlim}
\lim_{S \to \infty} \alpha_{\mathrm{L}}(S) = \lim_{S \to \infty} \alpha_{\mathrm{GP}}(S) = \alpha_{\textsc{BOV}}< 1.
\end{equation}
It is clear from Eq. \eqref{eq:alpha1} that
\begin{equation}
\alpha_{\mathrm{L}}(S) < \alpha_{\mathrm{L}}(S+1)
\end{equation}
and we similarly show in Proposition \ref{prop:ineq} that
\begin{equation}
\label{eq:monot}
\alpha_{\mathrm{GP}}(S) < \alpha_{\mathrm{GP}}(S+1).
\end{equation}
Thus both our proposed algorithms exhibit approximation ratios that strictly increase with $S$ and converge to $\alpha_{\textsc{BOV}}$ in the semiclassical limit. Finally, we quote explicit expressions for $\alpha_{\mathrm{GP}}(S)$ and $\alpha_{\textsc{BOV}}$ for completeness. Introducing the function~\cite{briet2010positive}
\begin{equation}
F^*(3,\rho) = \frac{8}{3\pi} \rho \, _2F_1(1/2,1/2;5/2;\rho^2),
\end{equation}
where $_2F_1(a,b;c;z)$ denotes the Gauss hypergeometric function, we have
\begin{equation}
\label{eq:alpha2def}
\alpha_{\mathrm{GP}}(S) = \min_{\rho \in [-1,0)} \frac{1 - F^*(3,\rho)}{1- \left(\frac{S+1}{S}\right)\rho}
\end{equation}
and
\begin{equation}
\label{eq:alphaBOVdef}
\alpha_{\textsc{BOV}} = \min_{\rho \in [-1,0)} \frac{1 - F^*(3,\rho)}{1-\rho}.
\end{equation}

These various approximation ratios are plotted as a function of $S$ in Fig. \ref{fig:approx}, which illustrates the inequalities Eq. \eqref{eq:ineq} and the limiting behaviour Eq. \eqref{eq:ineqlim}.

\section{Spin-$S$ Bri{\"e}t-Oliveira-Vallentin algorithm}
\label{sec:BOV}
We first note the following triplet of inequalities:
\begin{theorem}\emph{(Lieb, 1973~\cite{lieb1973classical})}
\label{thm:Lieb}
The values of the optimization problems $\textsc{CHA}$ and $\textsc{QHA}_S$ are related by the inequalities
\begin{equation}
\label{eq:Lieb}
\left(\frac{S+1}{S}\right)^2 \textsc{CHA}(G) \leq \textsc{QHA}_S(G) \leq  \textsc{CHA}(G).
\end{equation} 
\end{theorem}
These inequalities are implied by a result for ground states of spin systems first stated in Ref. \cite{lieb1973classical} that we refer to in this paper as ``Lieb's inequalities'', and are a special case of the Berezin-Lieb inequalities~\cite{lieb1973classical,berezin1975general}. The right inequality is immediate from the variational principle for products of Bloch coherent states; the left inequality is less trivial and requires that $\textsc{CHA}(G) \leq 0$, which is quickly seen by noting that for any non-trivial instance, $\max_{\vec{\mathbf{\Omega}}} H_{\textsc{CHA}}(G,\vec{\mathbf{\Omega}}) > 0$ but the expectation value of $H_{\textsc{CHA}}(G,\vec{\mathbf{\Omega}})$ under i.i.d. random uniform assignments of the $\mathbf{\Omega}_i$ is zero. Theorem \ref{thm:Lieb} yields the following approximation algorithm for $\textsc{QMaxCut}_S$:
\begin{corollary}
\label{cor:BOV}
Let $\{\mathbf{\Omega}_{i}\}_{i \in V}$ be an approximation to $\textsc{Prod(G)}$ obtained from the BOV algorithm. Then the product of spin-$S$ Bloch coherent states $|\vec{\mathbf{\Omega}}\rangle = \otimes_{i\in V} |\mathbf{\Omega}_i\rangle$ yields an 
\begin{equation}
\alpha_{\mathrm{L}}(S) = \left(\frac{S}{S+1}\right)^2 \alpha_{\textsc{BOV}}
\end{equation}
approximation to $\textsc{QMaxCut}_S(G)$.
\end{corollary}
\begin{proof}
We first note that the expectation value $\mathbb{E}_{Z}[\langle \mathbf{\Omega} | \hat{H}_{\textsc{QMC}_S}(G) |\mathbf{\Omega} \rangle]$ over randomized roundings $Z$ is precisely the value $\textsc{BOV}(G)$ of the BOV algorithm for approximating $\textsc{Prod}(G)$, which satisfies
\begin{equation}
\label{eq:BOV3MC}
\alpha_{\textsc{BOV}} \textsc{Prod}(G) \leq \textsc{BOV}(G) \leq \textsc{Prod}(G).
\end{equation}
Thus it suffices to show that \textsc{Prod} approximates \textsc{QMaxCut} with approximation ratio $\left(\frac{S}{S+1}\right)^2$, i.e. that
\begin{equation}
\label{eq:QMC3MC}
\left(\frac{S}{S+1}\right)^2 \textsc{QMaxCut}_S(G) \leq \textsc{Prod}(G) \leq \textsc{QMaxCut}_S(G).
\end{equation}
To see this, note that Theorem \ref{thm:Lieb} implies that $\textsc{CHA}$ approximates $\textsc{QHA}_S$ in the sense that
\begin{equation}
\label{eq:rearr}
\left(\frac{S}{S+1}\right)^2 |\textsc{QHA}_S(G)| \leq  |\textsc{CHA}(G)| \leq |\textsc{QHA}_S(G)|.
\end{equation}
Adding the sum of the weights $W = \frac{1}{2} \sum_{\{i,j\} \in E} w_{ij}$ to all terms of this inequality and using Eqs. \eqref{eq:qrelation} and \eqref{eq:crelation}, it follows that
 \begin{equation}
W + \left(\frac{S}{S+1}\right)^2 |\textsc{QHA}_S(G)| \leq  \textsc{Prod}(G) \leq \textsc{QMaxCut}_S(G).
\end{equation}
But by non-negativity of $W$ and Eq. \eqref{eq:qrelation},
\begin{equation}
\left(\frac{S}{S+1}\right)^2 \textsc{QMaxCut}_S(G) \leq W+\left(\frac{S}{S+1}\right)^2 |\textsc{QHA}_S(G)| 
\end{equation}
and Eq. \eqref{eq:QMC3MC} is immediate.
\end{proof}

To summarize, Lieb's inequalities imply via Corollary \ref{cor:BOV} that:
\begin{enumerate}
    \item The quantum optimization problem $\textsc{QMaxCut}_S$ is arbitrarily well approximated by a classical optimization problem as $S \to \infty$. 
    \item $\textsc{QMaxCut}_S$ admits polynomial-time approximation ratios that are arbitrarily close to the BOV approximation ratio $\alpha_{\textsc{BOV}}$ as $S \to \infty$.
\end{enumerate}
Note that for small $S$, for example $S=1/2$, our approximation ratio $\alpha_{\mathrm{L}}(S)$ is worse than the approximation ratio $\alpha = 0.25$ obtained from random guessing~\cite{GP19}, with $\alpha(S) = \alpha_{\textsc{BOV}}/9 \approx 0.109... < 0.25$. However, $\alpha_{\mathrm{L}}(S)$ performs better as $S$ increases, and once $S \approx 200$, our approximation ratio attains $99 \%$ of the BOV value.

This interpretation of Lieb's theorem as an approximability guarantee was made recently~\cite{bravyi2019approximation} in the context of approximating the local qubit Hamiltonian problem, for which $S=1/2$. In the next section, we propose an algorithm that yields a demonstrably better approximation ratio than predicted by Corollary \ref{cor:BOV}, even though it still rounds an SDP to a product of Bloch coherent states (see also Fig. \ref{fig:approx}). This raises the possibility (together with earlier findings~\cite{bravyi2019approximation} for $S=1/2$) that the lower bound in Lieb's inequality Eq. \eqref{eq:Lieb} may not be tight for any $S$.

\section{Spin-$S$ Gharibian-Parekh algorithm}
\label{sec:GP}
\subsection{The spin-$S$ SDP}
We now introduce a generalization of the Gharibian-Parekh algorithm~\cite{GP19} to spin $S>1/2$, as outlined in Section \ref{sec:SDP}. The first step is to derive the spin-$S$ SDP, Eq. \eqref{eq:SDPS}, which recovers suitable formulations~\cite{hwang2022unique,parekh2022optimal} of the Gharibian-Parekh SDP when $S=1/2$. In particular, we would like to show the following.
\begin{prop}
\label{prop:relax}
$\textsc{SDP}_S(G)$ as defined in Eq. \eqref{eq:SDPS} is a relaxation of $\textsc{QMaxCut}_S(G)$, i.e.
\begin{equation}
\textsc{SDP}_S(G) \geq \textsc{QMaxCut}_S(G).
\end{equation}
\end{prop}
\begin{proof}
Let $|\psi\rangle \in \mathbb{C}^{d^N}$ be an arbitrary $N$ qudit state and consider the properties that the $N$-by-$N$ matrix 
\begin{equation}
M_{ij} = \langle \psi | \hat{\mathbf{S}}_i \cdot \hat{\mathbf{S}}_j |\psi \rangle
\end{equation}
must satisfy. It follows by Hermiticity of the spin operators $\hat{S}_i^\alpha$ that $M$ is real, symmetric and positive semidefinite (PSD). Moreover, the spin-$S$ constraint Eq. \eqref{eq:Ssquared} sets the diagonal elements of $M$ to equal
\begin{equation}
M_{ii} = S(S+1).
\end{equation}
Since every state $|\psi\rangle$ defines such an $M$, it follows that
\begin{equation}
\label{eq:relax}
\max_{\substack{M \, \mathrm{real, \, PSD} \\ M_{ii}=S(S+1)}} \frac{1}{2}\sum_{\{i,j\} \in E}\left(1- \frac{1}{S^2} M_{ij}\right) \geq \textsc{QMaxCut}_S(G).
\end{equation}
To relate this to Eq. \eqref{eq:SDPS}, we define the matrix
\begin{equation}
\rho_{ij} = \frac{M_{ij}}{S(S+1)}
\end{equation}
which is also real and positive semidefinite, but is unit normalized with $\rho_{ii} = 1$. It follows that $\rho$ is the Gram matrix of some $\mathbf{y}_i \in S^{N-1}$, i.e.
\begin{equation}
\rho_{ij} = \mathbf{y}_i \cdot \mathbf{y}_j.
\end{equation}
Thus by Eq. \eqref{eq:relax},
\begin{align}
\max_{\substack{M \, \mathrm{real, \, PSD} \\ M_{ii}=S(S+1)}}\frac{1}{2}\sum_{\{i,j\} \in E}\left(1- \frac{1}{S^2} M_{ij}\right)
&= \max_{\mathbf{y}_i \in S^{N-1}} \frac{1}{2}\sum_{\{i,j\} \in E}\left(1- \left(\frac{S+1}{S}\right) \mathbf{y}_i \cdot \mathbf{y}_j\right) = \textsc{SDP}_S(G) 
\end{align}
is a relaxation of $\textsc{QMaxCut}_S$, as claimed.
\end{proof}

\subsection{The spin-$S$ approximation ratio}
We now establish the following result.
\begin{theorem}
\label{thm:GP}
Let $\{\mathbf{\Omega}_i\}_{i\in V}$ be the result of randomly rounding the output $\{\mathbf{y}_i\}_{i\in V}$ of the spin-$S$ SDP via Eq. \eqref{eq:round}. Then the product of spin-$S$ Bloch coherent states $|\vec{\mathbf{\Omega}}\rangle = \otimes_{i\in V} |\mathbf{\Omega}_i\rangle$ yields an $
\alpha_{\mathrm{GP}}(S)$ approximation to $\textsc{QMaxCut}_S$, where
\begin{equation}
\alpha_{\mathrm{GP}}(S) = \min_{\rho \in [-1,0)} \frac{1 - F^*(3,\rho)}{1- \left(\frac{S+1}{S}\right)\rho}.
\end{equation}
\end{theorem}
\begin{proof}
The proof follows the standard pattern~\cite{goemans1995improved,GP19} of such results. Write
\begin{equation}
W = \langle \vec{\mathbf{\Omega}} | \hat{H}_{\textsc{QMC}_S} | \vec{\mathbf{\Omega}} \rangle = \frac{1}{2}\sum_{\{i,j\}\in E} w_{ij}(1-\mathbf{\Omega}_i \cdot \mathbf{\Omega}_j).
\end{equation}
We first note that by the variational principle,
\begin{equation}
W \leq \textsc{QMaxCut}_S(G)
\end{equation}
so that the expectation value over randomized roundings
\begin{equation}
\mathbb{E}_Z[W] \leq \textsc{QMaxCut}_S(G).
\end{equation}
We next write $\rho_{ij} = \mathbf{y}_i \cdot \mathbf{y}_j$ and let $E^+ = \{\{i,j\} \in E: \rho_{ij} < \frac{S}{S+1}\} \subseteq E $ denote the set of edges on which the terms of $\textsc{SDP}_S(G)$ are positive. Then
\begin{align}
\nonumber 
\mathbb{E}_{Z}[W] &= \frac{1}{2}\sum_{\{i,j\}\in E} w_{ij}(1-\mathbb{E}_Z[\mathbf{\Omega}_i \cdot \mathbf{\Omega}_j]) \\
\nonumber 
&\geq \frac{1}{2}\sum_{\{i,j\}\in E^+} w_{ij}(1-\mathbb{E}_Z[\mathbf{\Omega}_i \cdot \mathbf{\Omega}_j]) \\
\nonumber 
&= \frac{1}{2}\sum_{\{i,j\}\in E^+} w_{ij}\left(\frac{1-F^*(3,\rho_{ij})}{1-\left(\frac{S+1}{S}\right)\rho_{ij}}\right) \left(1-\left(\frac{S+1}{S}\right) \rho_{ij}\right) \\
\nonumber 
&\geq \alpha_{\mathrm{GP}}(S) \cdot \frac{1}{2}\sum_{\{i,j\}\in E^+} w_{ij}\left(1-\left(\frac{S+1}{S}\right) \rho_{ij}\right) \\
\nonumber 
&\geq \alpha_{\mathrm{GP}}(S) \textsc{SDP}_S(G) \\
\label{eq:steps}
& \geq \alpha_{\mathrm{GP}}(S) \textsc{QMaxCut}_S(G),
\end{align}
where in the second line we used termwise non-negativity, in the third line we used the BOV result \cite{briet2010positive} for the expectation value $\mathbb{E}_Z[\mathbf{\Omega}_i \cdot \mathbf{\Omega}_j] = F^*(3,\rho_{ij})$, in the fourth line we defined
\begin{equation}
\alpha_{\mathrm{GP}}(S) =\min_{\rho \in [-1,\frac{S}{S+1})} \frac{1-F^*(3,\rho)}{1-\left(\frac{S+1}{S}\right)\rho},
\end{equation}
in the fifth line we restored all edges of $E$ in the summation, and in the last line we used Proposition \ref{prop:relax}. It remains to show that the minimization over $[0,\frac{S}{S+1})$ in the definition of $\alpha_{\mathrm{GP}}(S)$ is redundant. To this end, we define
\begin{equation}
\label{eq:fsdef}
f_S(\rho) = \frac{1-F^*(3,\rho)}{1-\left(\frac{S+1}{S}\right)\rho}
\end{equation}
and note that by standard properties~\cite{abramowitz1948handbook} of the hypergeometric function,
\begin{equation}
\label{eq:alphastarbd}
f_S(-1) = \frac{2S}{2S+1} < 1 = f_S(0).
\end{equation}
Thus minimizing over $\rho=0$ is redundant. Moreover for $\rho \in (0,1)$, since $_2F_1(1/2,1/2;5/2;\rho^2)$ is a power series in $\rho$ with positive coefficients, it follows that
\begin{equation}
0 < \, _2F_1(1/2,1/2;5/2;\rho^2) < \, _2F_1(1/2,1/2;5/2;1) = \frac{3\pi}{8}, \quad \rho \in (0,1),
\end{equation}
and therefore that
\begin{equation}
\label{eq:hypub}
F^*(3,\rho) < \rho < \left(\frac{S+1}{S}\right)\rho, \quad \rho \in (0,1).
\end{equation}
In particular,
\begin{equation}
f_S(\rho) > 1, \quad \rho \in (0,1).
\end{equation}
We deduce that the minimum of $f_S(\rho)$ over the interval $[-1,\frac{S}{S+1})$ must be attained for $\rho \in [-1,0)$ and the result follows.
\end{proof}
Finally, we prove the inequalities Eqs. \eqref{eq:ineq} and \eqref{eq:monot}.
\begin{prop}
\label{prop:ineq}
The inequalities Eqs. \eqref{eq:ineq} and \eqref{eq:monot} hold for all $S \in \{\frac{1}{2},1,\frac{3}{2},\ldots\}$.
\end{prop}
\begin{proof}
We first note that our definition of the BOV approximation ratio Eq. \eqref{eq:alphaBOVdef} is consistent with its usual formulation~\cite{briet2010positive} as a minimization over $[-1,1]$. This is obvious from the numerical value of the argmin~\cite{briet2010positive}, but for a proof let $g(\rho) = \frac{1-F^*(3,\rho)}{1-\rho}$ and note in this that case $g(-1)=g(0)=1$, while for $\rho \in (0,1)$, Eq. \eqref{eq:hypub} implies that $g(\rho)>1$. Finally we can omit the endpoint $\rho=1$ from the optimization by the steps leading to Eq. \eqref{eq:steps}.

Now fix an allowed value of $S$. The important point is that we can express both $\alpha_{\mathrm{GP}}(S)$ and $\alpha_{\textsc{BOV}}$ as minimization problems over the same interval $[-1,0)$, as in Eqs. \eqref{eq:alphaBOVdef} and \eqref{eq:alpha2def}, and it will suffice to establish the inequalities
\begin{equation}
\label{eq:simpineq}
\left(\frac{S}{S+1}\right)^2 g(\rho) < f_S(\rho) < f_{S+1}(\rho) < g(\rho), \quad \rho \in [-1,0),
\end{equation}
with $f_S(\rho)$ as defined in Eq. \eqref{eq:fsdef}. 

To show the middle and right-most inequalities in Eq. \eqref{eq:simpineq}, note that for $\rho \in [-1,0)$,
\begin{equation}
1-\rho < 1 - \left(\frac{S+2}{S+1}\right)\rho < 1 - \left(\frac{S+1}{S}\right)\rho,
\end{equation}
and
\begin{equation}
1-F^*(3,\rho) > 1.
\end{equation}
To show the left-most inequality in Eq. \eqref{eq:simpineq}, note that
\begin{align}
\nonumber &\frac{1-F^*(3,\rho)}{1 - \left(\frac{S+1}{S}\right)\rho} - \left(\frac{S}{S+1}\right)^2 \frac{1-F^*(3,\rho)}{1-\rho} \\
\nonumber =&\,\frac{1-F^*(3,\rho)}{(1-\rho)\left(1 - \left(\frac{S+1}{S}\right)\rho\right)}\left(1 - \left(\frac{S}{S+1}\right)^2 - \frac{\rho}{S+1}\right) \\
>&\,0
\end{align}
for $\rho \in [-1,0)$. The result follows upon minimizing Eq. \eqref{eq:simpineq} over $[-1,0)$ and noting that none of these functions attain their minima at $\rho=0$. Thus it suffices to perform this minimization over a closed subinterval $[-1,-\epsilon_S] \subseteq [-1,0)$ and Eq. \eqref{eq:simpineq} guarantees the result.
\end{proof}
Finally, we note that assuming the vector-valued Borell conjecture of Ref. \cite{hwang2022unique}, this algorithm achieves its integrality gap, defined by the ratio
\begin{equation}
\Delta_S = \inf_{G} \frac{\textsc{QMaxCut}_S(G)}{\textsc{SDP}_S(G)}
\end{equation}
and in this sense is optimal among algorithms that yield feasible solutions by rounding solutions to $\textsc{SDP}_S$. This follows by a straightforward extension of arguments provided in previous work~\cite{hwang2022unique}, based on constructing a specific high-degree instance (the ``Gaussian graph''~\cite{o2008optimal}) for which the state maximizing $\textsc{QMaxCut}_S$ is a product state by the results of Ref. \cite{brandao} and correspondingly analytically tractable.

\section{QMA-completeness}
We now consider the computational complexity of computing the value $\textsc{QMaxCut}_S(G)$ to additive inverse polynomial accuracy (rather than up to a constant multiplicative accuracy $\alpha$). To match the literature we consider the equivalent problem of computing $\textsc{QHA}_S(G)$, the minimum eigenvalue of the Heisenberg antiferromagnet $\hat{H}_{\textsc{QHA}_S}(G)$,

\begin{equation}
\hat{H}_{\textsc{QHA}_S}(G) = \frac{1}{2S^2}\sum_{\{i,j\}\in E} w_{ij} \hat{\mathbf{S}}_i \cdot \hat{\mathbf{S}}_j, \qquad w_{ij}\geq 0 \text{ for all } \{i,j\}\in E.
\label{eq:H_QHA}
\end{equation}

It was shown in \cite{piddock2021universal} that Hamiltonians of the form \eqref{eq:H_QHA} with mixed signs (i.e. without the condition that $w_{ij}\geq 0$) are universal in the sense that they can simulate any other local Hamiltonian with respect to the notion of analogue simulation introduced in \cite{Cubitt_2018}.
In particular this family of Hamiltonians can efficiently simulate those Hamiltonians for which the Local Hamiltonian Problem is QMA-complete, which implies that approximating the ground state energy for this family of Hamiltonians to inverse polynomial accuracy  is also QMA-complete.

The main result of this section is that Hamiltonians of the form $\hat{H}_{\textsc{QHA}_S}(G)$ are universal in the same sense, and hence that approximating $\textsc{QMaxCut}_S(G)$ to additive inverse polynomial accuracy is QMA-complete.

\begin{theorem}
    Spin-S antiferromagnetic Heisenberg Hamiltonians of the form $\hat{H}_{\textsc{QHA}_S}(G)$ are universal. 
    \label{thm:Suniversal}
\end{theorem}

\begin{corollary}
    Given $a<b \in \mathbb{R}$ and a weighted graph $G$ on $n$ vertices, with $b-a>1/\poly(n)$ and $w_{ij}<poly(n)$, it is QMA-complete to decide if $\textsc{QMaxCut}_S(G)$ is above $b$ or less than $a$.
\end{corollary}

To prove Theorem~\ref{thm:Suniversal}, it will suffice to show that antiferromagnetic spin-$S$ Heisenberg Hamiltonians can simulate spin-$S$ Heisenberg Hamiltonians with mixed signs, and then appeal to \cite{piddock2021universal}.

We do this using the well-established technique of perturbative gadgets \cite{oliveira2008complexity}, and ``mediator'' gadgets in particular. A mediator gadget is a Hamiltonian $\hat{H}$, which consists of a strongly weighted term $\Delta \hat{H}_0$ and other less strongly weighted terms. $\hat{H}_0$ acts non-trivially only on a set of qudits which we call ``mediator'' qudits, and has a non-degenerate ground state $\ket{\Psi_0}$ on this space. The parameter $\Delta$ is chosen to be large enough such that the low energy part of the total Hamiltonian $\hat{H}$ is approximately equal to $\hat{H}'\otimes \ket{\Psi_0}\bra{\Psi_0}$ for some effective $\hat{H}'$ acting on the non-mediator qudits.

The gadget we use is essentially the same as that used in 
\cite{piddock2015complexity} which covers the $S=1/2$ case, showing that qubit antiferromagnetic Heisenberg interactions can efficiently simulate Heisenberg interactions with mixed signs.
The idea of the basic gadget is to simulate a single ferromagnetic Heisenberg interaction $-\hat{\mathbf{S}}_1\cdot\hat{\mathbf{S}}_2$ between qudits 1 and 2, using only antiferromagnetic interactions. By repeating the gadget where necessary across the entire Hamiltonian, it is possible to simulate a Heisenberg Hamiltonian with mixed signs using only antiferromagnetic terms.
This gadget consists of two mediator qubits, $a,b,$ and the strongly weighted term is simply $\hat{H}_0=\S_a\cdot \S_b$, which has the non-degenerate singlet state $\ket{\Psi_0}$ as ground state. The less strongly weighted terms are of the form $\hat{H}_2=\S_1\cdot \S_a +\S_2\cdot \S_a$.
Then using second order perturbation theory, we can show that for large $\Delta$, the low energy part of $\hat{H}=\Delta \hat{H}_0+ \Delta^{1/2}\hat{H}_2$ is approximately equal to $-\hat{\mathbf{S}}_1\cdot\hat{\mathbf{S}}_2 \otimes \ket{\Psi_0}\bra{\Psi_0}$.

For full technical details, see Appendix~\ref{sec:gadgetproof}

\section{Conclusion}
We have introduced a family of spin-$S$ generalizations of Quantum Max-Cut and shown that these become progressively easier to approximate as $S \to \infty$, with a limiting approximation ratio given by the BOV approximation ratio. We have further shown that solving any of these optimization problems to inverse polynomial accuracy is QMA-complete. Our results thus demonstrate explicitly how a sequence of QMA-complete quantum optimization problems can converge ``in value'' to a classical optimization problem contained in NP. As a matter of physics alone, such convergence to a semiclassical limit is hardly surprising~\cite{lieb1973classical}. Augmented by the connections to computational hardness established here and elsewhere~\cite{piddock2021universal,kallaugher2024complexity}, our construction suggests a novel means of comparing quantum and classical hardness of approximation, which complements earlier related results~\cite{brandao,hwang2022unique}.

In future work, it would be desirable to develop approximation algorithms for spin-$S$ Quantum Max-Cut that go beyond the Gharibian-Parekh scheme, either by introducing entanglement between qudits~\cite{AGM20,king2022improved,lee2022optimizing} or by improving the product-state approximation ratio along the lines established by Parekh and Thompson for $S=1/2$~\cite{parekh2022optimal}. The latter result in particular relies on a certain nonlinear monogamy-of-entanglement inequality (Eq. 14 of Ref. \cite{parekh2022optimal}) that does not appear to generalize directly to higher spin $S > 1/2$. The emergence of classical Rank-3 Max-Cut from spin-$S$ Quantum Max-Cut as $S \to \infty$ suggests that any such attempts will be in tension with the conjectured~\cite{hwang2022unique} optimality of the BOV approximation ratio for large $S$. More generally, optimizing the approximation ratio for quantum or classical algorithms for spin-$S$ Quantum Max-Cut as $S \to \infty$ could provide an indirect route towards understanding the quantum or classical approximability of classical Rank-3 Max-Cut.

\section{Acknowledgements}
We thank D.A. Huse, R. King, S.L. Sondhi and J. Wright for conversations on related topics. VBB thanks the Simons Institute for the Theory of Computing for its hospitality during the initiation of this work.

\appendix

\section{Perturbative gadget proof}
\label{sec:gadgetproof}
Here we prove the claim that the Hamiltonian $\hat{H}=\Delta \hat{H}_0+ \Delta^{1/2}\hat{H}_2$ approximately simulates $-\hat{\mathbf{S}}_1\cdot\hat{\mathbf{S}}_2$.
where:
\[\hat{H}_0=(\S_a+\S_b)\cdot(\S_a+\S_b)= 2\S_a\cdot \S_b +2S(S+1)\]
 \[\hat{H}_2=\S_1\cdot \S_a +\S_2\cdot \S_a=(\S_1+\S_2)\cdot\S_a\]
To understand the low energy space of a Hamiltonian of the form of $\hat{H}$, Bravyi and Hastings \cite{bravyi2017complexity} developed simulation results at different orders of the perturbative expansion of the Schrieffer-Wolff transformation \cite{bravyi2011schrieffer}. These results were translated into the language of $(\Delta,\eta,\epsilon)$-simulations which underpin the definition of universality in \cite{Cubitt_2018}. We will need the second order variant (Lemma 5 from \cite{bravyi2017complexity}), which we restate here in Lemma~\ref{lem:secondorder}.

\begin{lem}[Bravyi-Hastings second-order simulation]
  \label{lem:secondorder}
  Let $\hat{H}_0$ be a Hamiltonian with ground state energy $0$ and spectral gap $\geq 1$. Let $\hat{P}$ be the projector onto the ground space of $\hat{H}_0$.
  Let $\hat{H}_1$, $\hat{H}_2$ be Hamiltonians acting on the same space, such that: $\max\{\|\hat{H}_1\|,\|\hat{H}_2\|\} \le \Lambda$; $\hat{P}\hat{H}_1(\hat{\mathbbm{1}}-\hat{P})=0$; and $\hat{P} \hat{H}_2 \hat{P}=0$.
  Suppose there exists a local isometry $\hat{V}$ such that $\hat{V}\hat{V}^{\dagger}=\hat{P}$ and 
  \begin{equation}
    \hat{V} \hat{H}_{\operatorname{target}} \hat{V}^\dag = \hat{P}\hat{H}_1\hat{P}- \hat{P}\hat{H}_2 \hat{H}_0^{-1} \hat{H}_2\hat{P}.
  \end{equation}
  where $\hat{H}_0^{-1}$ is the (Moore-Penrose) pseudo-inverse~\footnote{For a Hamiltonian with eigen-decomposition $\hat{H}=\sum_{\lambda}\lambda\ket{\psi_{\lambda}}\bra{\psi_{\lambda}}$, the pseudo-inverse is $\hat{H}_0^{-1}=\sum_{\lambda\neq0}\frac{1}{\lambda} \ket{\psi_{\lambda}}\bra{\psi_{\lambda}}$.
  }.
  Then $\hat{H}_{\operatorname{sim}} = \Delta \hat{H}_0 + \Delta^{1/2} \hat{H}_2 + \hat{H}_1$ $(\Delta/2,\eta,\epsilon)$-simulates $\hat{H}_{\operatorname{target}}$, provided that $\Delta \ge O(\Lambda^6/\epsilon^2 + \Lambda^2/\eta^2)$.
\end{lem}

To apply Lemma~\ref{lem:secondorder} and complete the proof, we will need some basic properties of the spin-$S$ Heisenberg interaction on two spins. We collect these standard results in Lemma~\ref{lem:SU2Heis} which can be found in e.g. \cite[Section 7]{piddock2021universal}.

\begin{lem}
\label{lem:SU2Heis}
    Let $\hat{H}_0$ be the Hamiltonian
    \[ \hat{H}_0=(\hat{\mathbf{S}}_a+\hat{\mathbf{S}}_b)\cdot(\hat{\mathbf{S}}_a+\hat{\mathbf{S}}_b) =2\hat{\mathbf{S}}_a\cdot \hat{\mathbf{S}}_b+2S(S+1)\] 
   \begin{enumerate}
       \item Then $\hat{H}_0$ is positive semi-definite and has non-degenerate ground state $\ket{\Psi_0}$ with eigenvalue 0, where:
       \[
    \ket{\Psi_0}=\frac{1}{\sqrt{2S+1}} \sum_{i \in \{-S,-S+1,\dots S-1,S\}}\ket{i}\ket{-i}\]
    \item The second smallest eigenvalue of $\hat{H}_0$ is 1.
    \item For $i \in \{1,2,3\}$, $\S_a^i\ket{\Psi_0}$ is in the $+1$ eigenspace of $\hat{H}_0$ 
    \item For $i,j \in \{1,2,3\}$,
    \[\bra{\Psi_0}\S^i_a\S^j_b\ket{\Psi_0}=\frac{1}{2S+1}\Tr \left(\S^i_a\S^j_b\right)=\delta_{ij}\frac{S(S+1)}{3(2S+1)}\]
   \end{enumerate}

    \end{lem}

Items 1 and 2 of Lemma~\ref{lem:SU2Heis} verify that $\hat{H}_0$ satisfies the conditions of Lemma~\ref{lem:secondorder}, with $\hat{P}=\ket{\Psi_0}\bra{\Psi_0}$.
Next we check that 
\[\hat{P}\hat{H}_2\hat{P}= \sum_{i=1}^3(\S_1^i+\S_2^i)\otimes \ket{\Psi_0}\bra{\Psi_0}\S_a^i \ket{\Psi_0}\bra{\Psi_0}=0\]
since $\bra{\Psi_0}\S_a^i \ket{\Psi_0}=\Tr(\S_a^i)/(2S+1)=0$.

Taking $\hat{V}$ to be the isometry that maps $\ket{\psi}\rightarrow \ket{\psi}\otimes \ket{\Psi_0}$, all that remains is to show that 
\begin{align}
    -\hat{P}\hat{H}_2 \hat{H}_0^{+} \hat{H}_2\hat{P}
    &=-\hat{P}(\S_1+\S_2)\cdot \S_a \hat{H}_0^{+}(\S_1+\S_2)\cdot \S_a \hat{P} \\
    &=-\sum_{i,j=1}^3 (\S_1^i+\S_2^i)(\S_1^j+\S_2^j) \otimes \hat{P}\S_a^i \hat{H}_0^{+} \S_a^j \hat{P} \\
    &=-\sum_{i,j=1}^3 (\S_1^i+\S_2^i)(\S_1^j+\S_2^j) \otimes \hat{P}\S_a^i  \S_a^j \hat{P} \\
    &=-\sum_{i,j=1}^3 (\S_1^i+\S_2^i)(\S_1^j+\S_2^j) \otimes \delta_{ij}\frac{S(S+1)}{3(2S+1)} \hat{P} \\
    &= -\frac{S(S+1)}{3(2S+1)}(\S_1+\S_2)\cdot (\S_1+\S_2) \otimes \hat{P} \\
    &= \left[-\frac{S(S+1)}{3(2S+1)}\S_1\cdot \S_2 - \frac{2S^2(S+1)^2}{3(2S+1)} \right]\otimes \hat{P}
\end{align}
where we used item 3 of Lemma~\ref{lem:SU2Heis} in the third line, and item 4 of Lemma~\ref{lem:SU2Heis} in the fourth line.
We can choose $\hat{H}_1=\frac{2S^2(S+1)^2}{3(2S+1)}\hat{\mathbbm{1}}$ if we wish to cancel out the identity term.

This shows how one mediator gadget simulates a single ferromagnetic spin-$S$ Heisenberg interaction. In order to simulate a general Heisenberg Hamiltonian with mixed signs, introduce a mediator gadget for each ferromagnetic interaction, and construct an overall $\hat{H}_0$ and $\hat{H}_2$ as the sum of the corresponding terms for each gadget. The antiferromagnetic terms can be included directly in $\hat{H}_1$, without any mediator spins. It is straightforward to check (see Lemma 17 of \cite{Cubitt_2018}) that there is no interference between the gadgets, so that for $\Delta$ sufficiently large, $\Delta \hat{H}_0 +\Delta^{1/2}\hat{H}_2+\hat{H}_1$ simulates the desired target Hamiltonian.

\bibliography{maxbib.bib}
\end{document}